\documentclass[11pt]{article}
\usepackage[margin=1in]{geometry}
\usepackage{amsmath,amssymb,amsthm}
\usepackage{hyperref}

\newtheorem{theorem}{Theorem}[section]
\newtheorem{proposition}[theorem]{Proposition}
\newtheorem{lemma}[theorem]{Lemma}
\newtheorem{definition}[theorem]{Definition}

\newcommand{\N}{\mathbb{N}}
\newcommand{\Syn}{\mathbf{Syn}}
\newcommand{\Exec}{\mathbf{Exec}}
\newcommand{\Beh}{\mathbf{Beh}}
\newcommand{\K}{K_U}
\DeclareMathOperator{\Run}{Run}

\title{Pushout Attachments and Conditional Complexity \\ of Executable Models}
\author{Alexander Kolpakov}
\date{\today}

\begin{document}

\maketitle

\begin{abstract}
We represent the modular enlargement of an executable model by a pushout of finite
typed presentations. Adhesivity preserves the original model and recovers
the shared interface, while chosen pushouts make linking functorial and
coproducts describe parallel attachments. A semantic comparison gives a
precise criterion for behavioral persistence. Effective linking also bounds
conditional prefix complexity; under the stated canonicalization and
invariance assumptions, attachment and extension codes yield the same profiles
of loss against complexity, up to additive constants, and share an exact
scalar envelope. The result separates syntax, semantics, ideal complexity,
and concrete coding cost, and shows that a preference for modularity can arise
only through restrictions on the permitted attachments.
\end{abstract}

\noindent\textbf{Keywords.}
Adhesive category; cofibration; executable model; Kolmogorov complexity;
minimum description length; pushout.

\smallskip
\noindent\textbf{Mathematics Subject Classification (2020).}\\
Primary 18B50; Secondary 18A30, 18A40, 68Q30, 68Q55.

\section{Introduction}

One often enlarges a finite executable presentation by joining a new module to
it along a typed interface.  We write this operation as the span
\[
 X\xleftarrow{u}A\xrightarrow{i}B
\]
and take its pushout \(P=X\sqcup_A B\).  The object \(X\) is the presentation
already in hand, \(A\) is the interface, and \(B\) is the incoming module.  If
the boundary map \(i\) is monic, then the induced map \(X\to P\) retains the
base as a subpresentation.

We place raw typed presentations in the presheaf slice
\(\Syn=\mathbf{Set}^{\mathcal C^{\mathrm{op}}}/T\), and denote its pointwise
finite objects by \(\Syn_{\mathrm f}\).  We call monomorphisms cofibrations.
The name merely marks the maps along which modules may be attached; it carries
no claim about weak equivalences, lifting axioms, or a model structure.
Presheaf slices are adhesive, so a pushout along such a map is also a pullback
and carries monomorphisms to monomorphisms~\cite{lack,toposes}.  The old syntax
therefore survives, and the overlap is exact.  None of this is new category
theory.  Its purpose here is to say, without metaphor, what survives when a
program grows by modules.

Yet the survival of syntax says nothing by itself about execution.  A checker
for a particular language may reject an ambient pushout, which is why we impose
admissibility separately.  Nor does successful compilation guarantee that old
behavior remains.  Given a semantics functor
\(\Run:\Exec\to\Beh\), applying \(\Run\) to an admissible attachment produces
a canonical comparison
\[
 \Run(X)\sqcup_{\Run(A)}\Run(B)\longrightarrow\Run(P).
\]
We call the attachment semantically exact when this map is invertible.  If
monomorphisms in \(\Beh\) are stable under pushout and \(\Run(i)\) is monic,
semantic exactness makes \(\Run(X)\to\Run(P)\) monic as well.  The comparison
map thus locates the proof that each application must supply, rather than
smuggling a semantic conclusion into a fact about syntax.

Fixing the base \(X\), we arrange attachment spans and pointed cofibration
extensions into two categories.  Chosen pushouts define a linking functor.  It
is left adjoint to the fully faithful functor that views an extension as an
attachment with no new interface.  Extensions are therefore reflective among
attachments.  The two effective constructions behind this reflection also
give the translations between codes used later.  After canonicalization, and
provided loss is unchanged by the relevant isomorphisms, the translations
shift the complexity profiles by at most constants.  If the grammar forbids
degenerate cells, however, the reflection may leave the grammar.  Only then
can the syntax of attachment express a real modular preference.

The point is the chain joining these observations.  An admissible attachment
along a monic boundary preserves the old syntax.  Under the stated hypotheses,
semantic exactness carries that persistence into behavior.  Effective linking
then turns a description of the attachment into a conditional description of
the extension.  The categorical ingredients and the scalar envelope are
familiar in isolation.  By placing the reflection and its coding consequence
side by side, we can use them without pretending that admissibility, semantic
preservation, ideal Kolmogorov complexity, and the length of a practical code
are interchangeable notions.

Gluing by pushouts connects the construction to adhesive graph transformation
~\cite{lack,toposes} and to categorical accounts of open systems based on
structured cospans~\cite{baezcourser}.  Structured cospans compose open systems
along their boundaries.  Here the base is fixed, an interface span is itself
an object, and the reflector forgets all but the linked pointed extension.
Our selection problem is a conditional restriction of algorithmic
loss--complexity and individual rate--distortion profiles
~\cite{vereshchagin,vereshchagin2010}, with minimum description length
supplying the interpretation in terms of actual codes~\cite{rissanen}.
Formalisms for executable world models motivate the discipline of
Section~2~\cite{schmidhuber1991,schmidhuber2004,steunebrink2013}, though the
categorical construction itself is not drawn from those works.

\section{Executable model presentations}

\subsection{The operational criterion}

\begin{definition}[Executable presentation]
Fix effective alphabets of actions \(\mathcal A\), observations \(\mathcal O\),
rewards \(\mathcal W\), patterns, values, and guards, together with a decidable
checker for types and interfaces and a single effective interpreter \(I\).  An
\emph{executable presentation} is a finite hierarchy \(M\) of guarded rules
that the checker accepts; its spaces of configurations and goals must also have
effective encodings.  The interpreter induces a computably enumerable
backward relation
\[
 R_M^-\subseteq
 \operatorname{Goal}_M\times\operatorname{Goal}_M^*
\]
and an effective relation for forward inference.  When actions are supplied as
conditions, the latter takes the form of the partial computable map
\[
 F_M:\operatorname{Conf}_M\times\mathcal A
       \rightharpoonup
       \operatorname{Conf}_M\times\mathcal O\times\mathcal W
\]
which may be iterated on each predicted configuration to produce a rollout.
The backward relation reads a consequent as a goal and its antecedents as
subgoals.  A rule body may call other named rules through typed references.
\end{definition}

Passing the checker certifies only syntactic executability.  It promises neither
termination nor usefulness.  An ill-posed input or a computation that never
halts has no output, and a semantic loss introduced later may assign it
infinite cost.

\subsection{A typed presheaf category}

Let \(\mathcal C\) be a finite, computably presented category.  Its objects
stand for syntactic sorts---rule nodes, antecedent and consequent ports,
guards, state slots, or references within a hierarchy---and its arrows record
incidence.  Fix a presheaf \(T\) of types and attributes.  We require effective
codes for every attribute that occurs in a finite object, so that equality,
typing, incidence, and naturality can all be decided from finite codes.  Raw
typed presentations then form the slice
\[
 \Syn=\mathbf{Set}^{\mathcal C^{\mathrm{op}}}/T.
\]
Write \(\Syn_{\mathrm f}\) for the full subcategory whose underlying presheaves
are pointwise finite.  Literal names and attributes that must be preserved are
encoded in the map to \(T\), rather than left as informal decorations.

A presentation discipline selects a subcategory
\(\Exec\subseteq\Syn_{\mathrm f}\) of well-formed finite open and closed
executable presentations; this subcategory need not be full.  Write
\(\iota:\Exec\hookrightarrow\Syn_{\mathrm f}\)
for its inclusion.
Its morphisms preserve types, the structure of guarded rules, ports, and
references within the hierarchy.  The discipline also specifies a semantics
functor
\[
 \Run:\Exec\longrightarrow\Beh,
\]
where \(\Beh\) may consist of transition systems, trace presentations, or any
other explicitly chosen kind of behavior.

The distinction matters.  Raw syntax has every finite colimit we shall need,
whereas a particular grammar, checker, compiler, or scheduler may not be
closed under arbitrary pushouts.  We therefore form all pushouts in
\(\Syn_{\mathrm f}\).  Such a pushout is \emph{admissible} when the four
objects and four arrows of its square all belong to \(\Exec\).  Since
\(\Exec\) need not be full, admissibility does not say that the same square is
a pushout there.

\section{Pushout attachments along cofibrations}

\subsection{Interface spans}

An attachment has the shape of a span
\[
 X\xleftarrow{u}A\xrightarrow{i}B.
\]
Its colimit identifies the two appearances of the interface \(A\); this
colimit is the pushout \(X\sqcup_A B\).

Throughout the paper, a \emph{cofibration} means a morphism of
\(\Syn_{\mathrm f}\) that is monic in \(\Syn\).  Isomorphisms are
cofibrations; composites and pushouts of cofibrations are again cofibrations.
The term has no further model-categorical force here.

\begin{lemma}[Finite typed pushouts]\label{lem:typed-pushouts}
The category \(\Syn\) is adhesive.  The subcategory \(\Syn_{\mathrm f}\) is
closed under finite colimits.  The class of its morphisms that are
monomorphisms in \(\Syn\) contains the isomorphisms, is closed under
composition, and is stable under pushout.
\end{lemma}

\begin{proof}
Every presheaf category is a topos, every slice of a topos is a topos, and
every topos is adhesive~\cite{lack,toposes}.  Hence a pushout along a
monomorphism is also a pullback, and the corresponding cobase change remains
monic.  Colimits in a presheaf category, and in its slices, are assembled
component by component from coproducts and quotients.  Since \(\mathcal C\)
has only finitely many objects, a finite colimit of pointwise finite
presheaves is still pointwise finite.  Closure under isomorphism and
composition is general; stability under pushout is the adhesive part of the
argument.
\end{proof}

\begin{definition}[Cofibration attachment]
Let \(X,A,B\in\Syn_{\mathrm f}\).  A \emph{cofibration attachment} consists of a
binding \(u:A\to X\) and a boundary inclusion \(i:A\hookrightarrow B\) that is
a cofibration.  Its linked presentation is
\[
 P=X\sqcup_A B.
\]
The attachment is \emph{executable} when this ambient pushout is admissible.
\end{definition}

In this picture, \(A\) records the typed dependencies of the new module,
\(B\) contains the module itself, and \(u\) binds the exposed ports to \(X\).
It is \(i\), not \(u\), that must be monic: this is precisely what makes the
cobase change preserve the base.

\begin{proposition}[Persistence and exact overlap]\label{prop:persistence}
For a cofibration attachment, the canonical map
\(j_X:X\to P\) is a cofibration, and the canonical comparison
\[
 A\longrightarrow X\times_P B
\]
is an isomorphism.  If the binding \(u\) is also a cofibration, then the
canonical map \(j_B:B\to P\) is a cofibration as well.
\end{proposition}

\begin{proof}
The map \(j_X\) is obtained by changing the base of \(i\) along \(u\), and is
therefore a cofibration.  In the ambient adhesive category, a pushout along a
monomorphism is a pullback.  Thus \(A\to X\times_P B\) is invertible.  If
\(u\) is itself a cofibration, the same argument applied on the other side
makes \(j_B\) a cofibration.
\end{proof}

The base is therefore still embedded, component by component, and the original
interface can be recovered as \(X\times_P B\).  When \(u\) is monic as well,
both pieces embed and meet in exactly \(A\).  Notice what the proposition does
not say.  It does not place \(P\) in \(\Exec\); that is the separate work of
admissibility.  Even after admissibility is known, the universal property still
lives in the ambient category unless the inclusion creates this pushout.  And
nothing so far guarantees that executions which once succeeded will continue
to do so.

\subsection{The category of attachments}

Fix \(X\in\Syn_{\mathrm f}\).  Let \(\mathbf{Ext}_X\) be the full
subcategory of the coslice \(X\mathbin{\downarrow}\Syn_{\mathrm f}\) whose
objects are cofibrations \(e:X\hookrightarrow Y\).  Let
\(\mathbf{Att}_X\) have objects
\[
 c=(X\xleftarrow{u}A\xrightarrow{i}B)
\]
with \(i\) a cofibration.  For
\(c'=(X\xleftarrow{u'}A'\xrightarrow{i'}B')\), a morphism
\((a,b):c\to c'\) consists of maps \(a:A\to A'\) and \(b:B\to B'\)
satisfying \(u=u'a\) and \(bi=i'a\).

\begin{theorem}[Attachment reflection]\label{thm:attachment-reflection}
After choosing pushouts, linking defines a functor
\[
 L_X:\mathbf{Att}_X\longrightarrow\mathbf{Ext}_X,
 \qquad
 (X\xleftarrow{u}A\xrightarrow{i}B)
 \longmapsto (X\hookrightarrow X\sqcup_A B).
\]
The functor
\[
 D_X:\mathbf{Ext}_X\longrightarrow\mathbf{Att}_X,
 \qquad
 (j:X\hookrightarrow Y)\longmapsto
 (X\xleftarrow{\mathrm{id}_X}X\xrightarrow{j}Y)
\]
is fully faithful and \(L_X\dashv D_X\).  In particular, the counit
\(L_XD_X\Rightarrow\mathrm{Id}_{\mathbf{Ext}_X}\) is an isomorphism, so
\(\mathbf{Ext}_X\) is reflective in \(\mathbf{Att}_X\).
\end{theorem}

\begin{proof}
A morphism \((a,b):c\to c'\) induces a unique map
\(L_Xc\to L_Xc'\) under \(X\), by the universal property of the pushout.
Uniqueness gives functoriality, and Proposition~\ref{prop:persistence} says
that the source map of every linked extension is a cofibration.  For
\(e=(j:X\hookrightarrow Y)\), pushout universality gives natural bijections
\[
 \mathbf{Ext}_X(L_Xc,e)
 \cong\{b:B\to Y\mid bi=ju\}
 \cong\mathbf{Att}_X(c,D_Xe).
\]
Thus \(L_X\dashv D_X\).  For a morphism between degenerate attachments, the
map on interfaces can only be \(\mathrm{id}_X\).  It follows that \(D_X\) is
fully faithful; more explicitly, it sends a morphism \(h\) in the coslice to
\((\mathrm{id}_X,h)\).  Finally, the pushout of
\(X\xleftarrow{\mathrm{id}_X}X\xrightarrow{j}Y\) is canonically \(Y\).
This is the counit isomorphism.
\end{proof}

The reflection belongs, for now, to the ambient syntax category.  It descends
to attachments inside \(\Exec\) only if the chosen pushouts, together with
all their induced mediating maps, remain within the presentation discipline.

\begin{proposition}[Parallel composition]\label{prop:parallel}
The category \(\mathbf{Ext}_X\) has finite coproducts, created by the coslice.
For \(e_s:X\hookrightarrow Y_s\), their coproduct has apex the iterated
pushout \(Y_1\sqcup_X\cdots\sqcup_XY_m\).  If each \(e_s=L_Xc_s\), this apex
is canonically the colimit of the corresponding finite star of attachment
spans.  Thus parallel attachments are independent of ordering up to the
canonical coproduct isomorphisms.
\end{proposition}

\begin{proof}
The identity on \(X\) is initial in \(\mathbf{Ext}_X\), while binary
coproducts in the coslice are pushouts over \(X\).  Their common base map
factors as
\(X\hookrightarrow Y_1\hookrightarrow Y_1\sqcup_XY_2\); the second map is a
base change of \(X\hookrightarrow Y_2\).  Since both factors are
cofibrations, the coproduct again lies in \(\mathbf{Ext}_X\); fullness means
that the coslice coproduct is created there.  Iterating gives all finite
coproducts.  The last assertion follows because the iterated pushout and the
star colimit have the same universal property.
\end{proof}

An incoming cell may, of course, depend on cells already attached.  Given
\[
 X_{k-1}\xleftarrow{u_k}A_k\xrightarrow{i_k}B_k,
 \qquad
 X_k=X_{k-1}\sqcup_{A_k}B_k,
\]
with each \(i_k\) a cofibration, Proposition~\ref{prop:persistence}, applied
at every stage, gives the filtration
\[
 X_0\hookrightarrow X_1\hookrightarrow\cdots\hookrightarrow X_n.
\]
Every arrow in this filtration is a cofibration, and every finite stage remains
in \(\Syn_{\mathrm f}\).

\section{Execution semantics}

\subsection{The comparison map}

Take an executable attachment, and suppose that the following pushout of
behaviors exists:
\[
 Q=\Run(X)\sqcup_{\Run(A)}\Run(B).
\]
The syntactic pushout cocone, after applying \(\Run\), gives compatible maps
from \(\Run(X)\) and \(\Run(B)\) into \(\Run(P)\).  The universal property of
\(Q\) therefore supplies a canonical comparison
\[
 \chi:Q\longrightarrow\Run(P).
\]

\begin{definition}[Semantically exact attachment]
An executable attachment is \emph{semantically exact} when \(\chi\) is an
isomorphism.
\end{definition}

Exactness is one sufficient compatibility condition for a language and its
linker.  The syntactic pushout does not imply it, nor need it be necessary for
every other reasonable notion of preserving behavior.

\begin{theorem}[Semantic monicity criterion]\label{thm:behavioral-persistence}
Assume that the displayed behavioral pushout exists, monomorphisms in
\(\Beh\) are stable under pushout, \(\Run(i):\Run(A)\to\Run(B)\) is monic,
and the attachment is semantically exact.  Then
\[
 \Run(j_X):\Run(X)\longrightarrow\Run(P)
\]
is monic.
\end{theorem}

\begin{proof}
The canonical map \(\Run(X)\to Q\) is obtained by pushing out \(\Run(i)\), so
it is monic.  Composing it with the isomorphism
\(\chi:Q\to\Run(P)\) gives \(\Run(j_X)\), which is therefore monic.
\end{proof}

The conclusion is internal to the chosen category \(\Beh\).  It says that old
traces or transitions are literally preserved only if morphisms in that
category also reflect traces or transitions.  Along any finite chain of
attachments satisfying the theorem, \(\Run(X_0)\to\Run(X_n)\) is monic because
monomorphisms are closed under composition.

The assumptions separate two ways the construction can fail.  First, a guard
or a clash of names may cause the compiler to reject the raw pushout.  Second,
even a program that compiles may acquire interactions absent from the pushout
of its component behaviors, in which case \(\chi\) is not invertible.  Suppose,
for example, that a global scheduler always chooses the first matching rule.
Attaching a new rule of higher priority can then suppress an output formerly
chosen from \(X\), even though the pushout of rule graphs is admissible.  On
the proposed syntax morphisms, such a scheduler may not define a functor at
all.  Passing to a coarser category of behaviors can restore functoriality, but
the comparison map may still fail to be invertible.  This elementary example
is enough to show why syntax alone cannot preserve behavior.  Enumerating
traces proves exactness only when the declared semantics is finite, effectively
presented, and exhaustively checkable.  A finite sample never proves global
equivalence of programs.

\section{Conditional complexity of linking}

Fix a universal prefix machine \(U\), together with effective self-delimiting
codes for finite presentations, morphisms, and tuples.  We write
\(\K(z\mid w)\) for conditional prefix Kolmogorov complexity.  A candidate is
encoded as a pointed extension \((X,j,Y)\), not merely as a labeled object in
isolation, and the base \(X\) is given a fixed canonical code.  To
canonicalize an extension in the coslice
\(X\mathbin{\downarrow}\Syn_{\mathrm f}\), first represent \(j\) as a
literal inclusion and keep the labels on its image.  At each sort, rename the
remaining elements by a standard finite set, and choose the least among the
finitely many incidence codes that result.  Renaming must preserve both the
map to \(T\) and the marked arrow.  Our decidability assumptions on finite
attributes make this procedure effective.  The resulting complexity still
depends, as it must, on the chosen representation and universal machine, but
not on pointless changes of names outside the base.

For an attachment to \(X\), let
\[
 c=(A,B,u,i)
\]
be its complete attachment code.  Componentwise finite coproducts, quotienting
by the interface identifications, and canonical relabeling give a fixed
computable linker \(\operatorname{Link}\), total on valid finite span codes,
with literal domain \(X\) and
\[
 \operatorname{Link}(X,c)=(j_X:X\hookrightarrow P),
 \qquad P=X\sqcup_A B,
\]
after canonicalization.  The linker constructs raw syntax; it is not asked to
decide whether the result belongs to \(\Exec\).  The following bounds concern
those outputs that are admissible.

\begin{proposition}[Complexity bound for effective linking]\label{prop:gluing-bound}
There is a constant \(b\), depending only on the encodings, universal machine,
and linker, such that every admissible attachment satisfies
\[
 \K((j_X:X\hookrightarrow X\sqcup_A B)\mid X)
 \le \K(A,B,u,i\mid X)+b.
\]
If the cell \(B\) is already available as conditional information, then
\[
 \K((j_X:X\hookrightarrow X\sqcup_A B)\mid X,B)
 \le \K(A,u,i\mid X,B)+b.
\]
\end{proposition}

\begin{proof}
A fixed program first runs a shortest conditional description of the
attachment and then calls \(\operatorname{Link}\).  Its own length is constant.
For the second inequality, the same program receives \(B\) on its conditional
tape.
\end{proof}

This is the usual monotonicity of prefix complexity under a fixed partial
computable map, here applied to the linker.  In the first inequality the cell
\(B\) is not free information; in the second it is, because both sides receive
it as a condition.  The same care is needed with libraries.  If every
candidate may call one fixed library \(R\), then every code and every
complexity must be conditioned on the same pair \((X,R)\).  Giving each
candidate its own uncharged library destroys the comparison.  Finally, these
statements are upper bounds, not additive identities: shared algorithmic
information can make a joint description shorter than the sum of its parts.

For a finite sequence of attachment codes \(c_{1:n}\), iterated linking is
also a fixed partial computable operation.  Writing \(j_{0,n}:X_0\hookrightarrow
X_n\) for the composite pointed extension, one has
\[
 \K((j_{0,n}:X_0\hookrightarrow X_n)\mid X_0)
 \le \K(c_{1:n}\mid X_0)+O(1).
\]
If \(\ell(c_{1:n}\mid X_0)\) is the length of a declared self-delimiting code
with a fixed uniform effective decoder, then
\[
 \K((j_{0,n}:X_0\hookrightarrow X_n)\mid X_0)
 \le \ell(c_{1:n}\mid X_0)+O(1).
\]
The last inequality is a rule for practical bookkeeping, not an exact
calculation of \(\K\).  Once the self-delimiting decoder for tuples has been
fixed, its \(O(1)\) term is uniform in \(n\).

\section{Choosing extensions by Kolmogorov complexity}

\subsection{Conditional profiles of loss and complexity}

Once a universal machine has been fixed, Kolmogorov complexity
~\cite{kolmogorov} becomes a coordinate measuring description.  For finite
data \(x\), the classical structure function is
\[
 h_x(\alpha)
 =\min\{\log_2|S|:x\in S,\ S\text{ finite},\ \K(S)\le\alpha\},
\]
with value \(+\infty\) when the constraint set is empty
~\cite{vereshchagin,livitanyi}.  Algorithmic rate--distortion functions
generalize this expression by replacing \(\log |S|\) with a distortion criterion
~\cite{vereshchagin2010}.

We now consider only candidates that extend a given base.  Fix a
presentation \(X\in\Exec\), finite data \(D\), and a countable, effectively
presented class \(\mathcal E_X\) of morphisms
\(e=(j:X\hookrightarrow Y)\) in \(\Exec\) whose underlying maps are
cofibrations.  The marked arrow belongs to the candidate, since one target may
extend \(X\) in several different ways.  Let
\[
 L_D(e)\in[0,+\infty]
\]
be the chosen loss under execution.  We allow infinite loss when execution
does not terminate or a semantic obligation fails, and write
\(k_X(e)=\K(e\mid X)\).

\begin{definition}[Profiles]
For \(\alpha\in\N\) and finite \(r\ge0\), the two conditional
profiles are
\[
 H_{D,X}(\alpha)
 =\min\{L_D(e):e\in\mathcal E_X,\ k_X(e)\le\alpha\}
\]
and
\[
 C_{D,X}(r)
 =\min\{k_X(e):e\in\mathcal E_X,\ L_D(e)\le r\},
\]
Each minimum is understood to be \(+\infty\) if its constraint set is empty.
\end{definition}

For fixed \(\alpha\), only finitely many candidate codes have a prefix
description of length at most \(\alpha\).  Thus the first minimum, if finite,
is attained.  The second is attained whenever finite because complexity takes
integer values.
Consequently, for \(\alpha\in\N\) and finite \(r\),
\[
 H_{D,X}(\alpha)\le r
 \quad\Longleftrightarrow\quad
 C_{D,X}(r)\le\alpha.
\]
In general these profiles are ideal, uncomputable objects.  They reduce to the
classical structure function when \(X\) is empty or fixed, the candidates are
finite sets \(S\) containing \(D\), and \(L_D(S)=\log_2|S|\), apart from the
fixed convention used to encode candidates.

\subsection{The scalar envelope and MDL specialization}

For \(\lambda>0\), define
\[
 \Phi_{D,X}(\lambda)
 =\inf_{e\in\mathcal E_X}
   \bigl(L_D(e)+\lambda k_X(e)\bigr).
\]
As usual, an infimum over the empty set is \(+\infty\).

\begin{proposition}[Exact scalar envelope]\label{prop:envelope}
For every \(\lambda>0\),
\[
 \Phi_{D,X}(\lambda)
 =\inf_{\alpha\in\N}
   \bigl(H_{D,X}(\alpha)+\lambda\alpha\bigr)
 =\inf_{r\ge0}
   \bigl(r+\lambda C_{D,X}(r)\bigr).
\]
If at least one candidate has finite loss, the infimum defining
\(\Phi_{D,X}(\lambda)\) is attained, and every minimizer is Pareto-optimal in
the two coordinates \((L_D,k_X)\).
\end{proposition}

\begin{proof}
For any candidate \(e\), set \(\alpha=k_X(e)\).  Then
\(H_{D,X}(\alpha)\le L_D(e)\), so the first right-hand infimum is at most
\(\Phi_{D,X}(\lambda)\).  Conversely, a finite value of
\(H_{D,X}(\alpha)\) is realized by some \(e_\alpha\), and
\[
 \Phi_{D,X}(\lambda)
 \le L_D(e_\alpha)+\lambda k_X(e_\alpha)
 \le H_{D,X}(\alpha)+\lambda\alpha.
\]
Taking infima proves the first equality.  The second is the same argument with
the coordinates exchanged: take \(r=L_D(e)\) in one direction, and in the
other choose a candidate that realizes \(C_{D,X}(r)\).

To prove attainment, choose a candidate whose objective has some finite value
\(M\).  Any candidate doing at least as well must satisfy
\(k_X(e)\le M/\lambda\), and there are only finitely many such candidates.
One of them therefore minimizes the objective.  If another candidate improved
both coordinates weakly and one of them strictly, it would strictly lower the
scalar objective.  Every minimizer is consequently Pareto optimal.
\end{proof}

For a probabilistic instance, suppose
\(\{P_e(\,\cdot\mid X):e\in\mathcal E_X\}\) is a family of conditional
semimeasures that is uniformly lower semicomputable from \((e,X)\), and set
\(-\log_2 0=+\infty\).  The coding inequality gives
\[
 \K(D\mid e,X)
 \le-\log_2 P_e(D\mid X)+O(1),
\]
and a self-delimiting concatenation therefore gives
\[
 \K(D\mid X)
 \le \K(e\mid X)-\log_2 P_e(D\mid X)+O(1),
\]
where the constant is uniform in \(D,e,X\) and depends only on the decoder for
the family.  Taking \(L_D(e)=-\log_2P_e(D\mid X)\), we find that
\(\Phi_{D,X}(1)\) is the ideal conditional optimum for two-part MDL.  It is
finite precisely when at least one candidate gives \(D\) positive mass
~\cite{rissanen,livitanyi}.  The conclusion depends on the uniform coding
hypothesis just stated; it does not follow for an arbitrary real-valued loss.

\subsection{Extension codes versus attachment codes}

Let \(\mathcal C_X\) be the valid codes \(c=(A,B,u,i)\) of admissible
cofibration attachments for which \(e_c=\operatorname{Link}(X,c)\) lies in
\(\mathcal E_X\).  Assume that \(\mathcal E_X\) remains closed when we apply
the fixed canonicalization that preserves \(X\), and that pointed isomorphisms
fixing \(X\) do not change the loss.  Define the ideal profile and scalar
objective for attachment codes by
\[
 H^{\mathrm{att}}_{D,X}(\alpha)
 =\min\{L_D(e_c):c\in\mathcal C_X,\ \K(c\mid X)\le\alpha\}
\]
and
\[
 \Phi^{\mathrm{att}}_{D,X}(\lambda)
 =\inf_{c\in\mathcal C_X}
   \bigl(L_D(e_c)+\lambda\K(c\mid X)\bigr).
\]

For this comparison, the object map of \(L_X\) is the one canonical computable
linker already fixed above, used uniformly for every \(X\).  The construction
\((X,e)\mapsto D_Xe\) is uniform and computable as well.

\begin{theorem}[Comparing attachment and extension codes]\label{thm:attachment-profile}
There are nonnegative integer constants \(b_1,b_2\), depending only on the
fixed machine, encodings, linker, and canonicalizer and independent of
\(D,X,\alpha,\lambda\), such that
\[
 H_{D,X}(\alpha+b_1)
 \le H^{\mathrm{att}}_{D,X}(\alpha),
 \qquad
 \Phi_{D,X}(\lambda)
 \le\Phi^{\mathrm{att}}_{D,X}(\lambda)+\lambda b_1,
\]
and
\[
 H^{\mathrm{att}}_{D,X}(\alpha+b_2)
 \le H_{D,X}(\alpha),
 \qquad
 \Phi^{\mathrm{att}}_{D,X}(\lambda)
 \le\Phi_{D,X}(\lambda)+\lambda b_2.
\]
\end{theorem}

\begin{proof}
For \(e=(j:X\hookrightarrow Y)\in\mathcal E_X\), the degenerate attachment
\(D_Xe=(X\xleftarrow{\mathrm{id}_X}X\xrightarrow{j}Y)\) is admissible: its
pushout square has maps \(\mathrm{id}_X,j,j,\mathrm{id}_Y\) in \(\Exec\).
Its canonical link is isomorphic to \(e\) as a pointed extension.  Closure
under canonicalization therefore places the link in \(\mathcal E_X\), and the
code in \(\mathcal C_X\).
The fixed effective object maps of \(L_X\) and \(D_X\) give constants
\[
 \K(L_Xc\mid X)\le\K(c\mid X)+b_1,
 \qquad
 \K(D_Xe\mid X)\le\K(e\mid X)+b_2.
\]
The first inequality is Proposition~\ref{prop:gluing-bound}; the second comes
from the fixed construction of a degenerate attachment.  By
Theorem~\ref{thm:attachment-reflection},
\(L_XD_Xe\cong e\) under \(X\),
and invariance under pointed isomorphism preserves its loss.  Apply the two
complexity bounds to the feasible sets and then to the scalar objectives.  The
four claimed inequalities follow.
\end{proof}

The converse is allowed to use degenerate cells, and this is essential.  It
says that an unrestricted attachment language has no more inductive bias than
the extension class itself, except for constants introduced by coding.  A
restricted grammar of primitive cells can express a genuine preference for
modularity.  In that case the converse inequalities may fail, though the
forward bounds supplied by the linker survive.

\subsection{Machine dependence and computable proxies}

\begin{proposition}[Invariance]
Hold \(\mathcal E_X\) and \(L_D\) fixed, and let \(U\) and \(V\) be optimal
universal prefix machines.  There is a nonnegative integer \(c\) such that,
for all \(D,X\), \(\alpha\in\N\), finite \(r\ge0\), and \(\lambda>0\),
\[
 H^{U}_{D,X}(\alpha+c)\le H^{V}_{D,X}(\alpha),
 \qquad
 H^{V}_{D,X}(\alpha+c)\le H^{U}_{D,X}(\alpha),
\]
\[
 C^{U}_{D,X}(r)\le C^{V}_{D,X}(r)+c,
 \qquad
 C^{V}_{D,X}(r)\le C^{U}_{D,X}(r)+c,
\]
and, whenever the scalar optima are finite,
\[
 |\Phi^{U}_{D,X}(\lambda)-\Phi^{V}_{D,X}(\lambda)|
 \le\lambda c.
\]
\end{proposition}

\begin{proof}
The invariance theorem supplies a constant \(c\) for which
\(|K_U(e\mid X)-K_V(e\mid X)|\le c\), uniformly in \(e\) and \(X\).  Insert
this bound into the constraints defining the two profiles to obtain the first
four inequalities.  Insert it into the scalar objectives and take infima to
obtain the last.
\end{proof}

The invariance constant does not make a choice from finite data canonical.  A
selector that can actually be run must still fix a machine, or at least a
computable prefix code.  Let
\(\mathcal C_X^{\mathrm{code}}\) be a declared family of valid codes with one
uniform prefix decoder, and let \(\ell(c\mid X)\) be the explicit code length.
It may minimize
\[
 L_D(\operatorname{Link}(X,c))+\lambda\ell(c\mid X),
 \qquad c\in\mathcal C_X^{\mathrm{code}}.
\]
If the list of codes is explicit and finite, the linker is total, and the loss
is total and rational valued, exhaustive evaluation decides the minimum.  For
general computable reals one may obtain only certified intervals, unless a gap
separating the candidates is known.  With an enumerable family or partial
execution, a search that may be stopped at any time cannot in general certify
that its current answer is globally best.

An explicit code length bounds the prefix complexity of the linked extension,
up to the constant cost of its decoder.  But it scores derivations, not
isomorphism classes.  Several distinct codes may produce the same extension,
and their multiplicity changes the induced prior.  Without further assumptions,
one therefore gets no general guarantee of approximation to the ideal
Kolmogorov frontier.

\section{Conclusion}

An admissible pushout along a monic boundary keeps the base intact and recovers
the interface as the exact place where old and new syntax meet.  The attachment
reflection organizes these pushouts into a functor: pointed cofibration
extensions sit reflectively inside attachment spans, and finite coproducts
describe modules linked in parallel.  Semantic exactness is the extra
condition that carries monicity from syntax into behavior.  Effective linking
and the fully faithful construction of degenerate attachments translate codes
in both directions.  Unrestricted attachment codes and extension codes thus
trace the same profiles of loss against complexity, up to uniform additive
constants, while the scalar envelope turns those profiles into a regularized
choice.

The construction is deliberately conditional.  For a concrete executable
language, one must still prove that the chosen pushouts are admissible and that
the semantics is exact, or at least conservative in the required sense.  A
meaningful preference for modularity also demands a restricted grammar of
primitive cells, one that excludes the degenerate attachments in the image of
\(D_X\).  Those two tasks are where the character of a particular language,
and the substance of a particular model-selection argument, finally enter.


{\small
\bibliographystyle{plain}
\bibliography{inverse_colimits_executable_models}
}

\end{document}